\documentclass[12pt,onecolumn,draftcls]{IEEEtran} 

\usepackage{verbatim}
\usepackage{amssymb}
\usepackage{amscd}
\usepackage{amsmath}
\usepackage{graphicx}
\usepackage{epstopdf}
\usepackage{epsfig}
\usepackage{amsthm}
\usepackage{picture}

\newtheorem{remark}{Remark}
\newtheorem{theorem}{Theorem}[section]

\newtheorem{proposition}{Proposition}[section]

\begin{document}
\title{{\footnotesize\tt Submitted to the 51$^{\rm st}$ IEEE Conf. on Decision and Control,  Dec. 10-13, 2012}\\
Animal-Inspired Agile Flight Using Optical Flow Sensing}
\author{K.\ Sebesta \& J.\ Baillieul
\thanks{Research support by ODDR\&E MURI10 Program Grant Number N00014-10-1-0952 is gratefully acknowledged.}}
\IEEEaftertitletext{\vspace{-2\baselineskip}} 
\maketitle

\begin{abstract}
There is evidence that flying animals such as pigeons, goshawks, and bats use optical flow  sensing to enable high-speed flight through forest clutter.  This paper discusses the elements of a theory of controlled flight through obstacle fields in which motion control laws are based on optical flow sensing.  Performance comparison is made with feedback laws that use distance and bearing measurements, and practical challenges of implementation on an actual robotic air vehicle are described.  The related question of fundamental performance limits due to clutter density is addressed.
\end{abstract}


\section{INTRODUCTON}\setcounter{equation}{0}
\label{sec:jb:Intro}

In the natural world, one finds enormous diversity in visual capabilities among different species of animals.  Birds of prey typically possess significantly greater visual acuity than humans (or other animals), and while most mammals are inferior to humans in visual acuity, they have superior abilities in detecting motion.  In the present paper, possible mechanisms by which certain flying animals use visual sensory feedback to control their movement though cluttered environments are discussed.  The principal emphasis will be on the use of optical flow techniques.

Although there is growing evidence that flying animals use optical flow sensing to navigate and control their motions (\cite{Huston}), the ways in which such sensing is integrated with other sensory modalities (inertial sensing, stereo visual depth perception, etc.) is not well understood.  Flying animals tend to have constrained visual capabilities---e.g. immobile eyes with fixed-focus optics.  These animals must therefore employ visual strategies that are adapted to their specialized visual abilities.  Many of these strategies emphasize visual cues derived from motions of retinal images.  (See \cite{Srinivasan}.)

The interplay between motion control and motion perception has recently been studied by a number of researchers.  Some of this work has been inspired by animal predators that have developed the ability to camouflage their motions when approaching to capture prey.  The key to this strategy is to  control the apparent relative motion as perceived on the retina of the prey.  (\cite{Mizutani})  In the spirit of the research reported below, Justh and Krishnaprasad (\cite{Justh}) have proposed a biologically plausible feedback flight control law that exhibits motion camouflage relative to a moving target.

The present paper describes some of our recent work aimed at developing robotic motion control control laws that produce motions resembling those of animals.  Of particular interest is high-speed flight through forest clutter.  There are a number of fundamental research questions that the work addresses.  While different flying animals employ complex combinations of sensing modalities, a fundamental premise of the paper is that in idealized environments, it should be possible to successfully use optical feedback alone to guide a UAV through an obstacle field.  The research has the following components: 1. We are continuing the development of operating regime dependent motion primitives (\cite{Baillieul04}) that will create a set of behaviors that is rich enough to achieve our goal of high-speed flight through a field of obstacles.  2. We would like to use optical sensory feedback in a way that reflects what is known about the mechanisms of animal flight and in a way that the motions that are generated by the proposed laws have qualitative and quantitative characteristics of animal motions observed in nature.  3. Finally, we would like to understand implementation independent fundamental bounds on the performance limits of flight through models of forest and other types of clutter.  Recent work by Karaman and Frazzoli \cite{frazzoli} has shown that under a certain idealization there is a critical velocity below which a mathematical model of a bird can fly through a model forest without colliding with a tree yet above this critical velocity, an infinitely long collision-free flight path does not exist.  Our goal is to understand how the answers to such questions change as we vary our optics-enabled control designs.

\subsection{Interpreting optical flow generated by thin objects with sharp boundaries}

Perhaps the single-most important feature of optical flow sensing is that when a vehicle is moving toward an object, it enables a very simple determination of {\em time-to-contact} under constant closing velocity.  This determination can be made without knowledge of the size, distance, or velocity of motion toward the object.  To see this transparently, we consider the idealized thin obstacle depicted in Fig.\ \ref{fig:jb:OpticalFlow}.
\begin{figure}[ht] 
\begin{center}\includegraphics[width=2.3in]{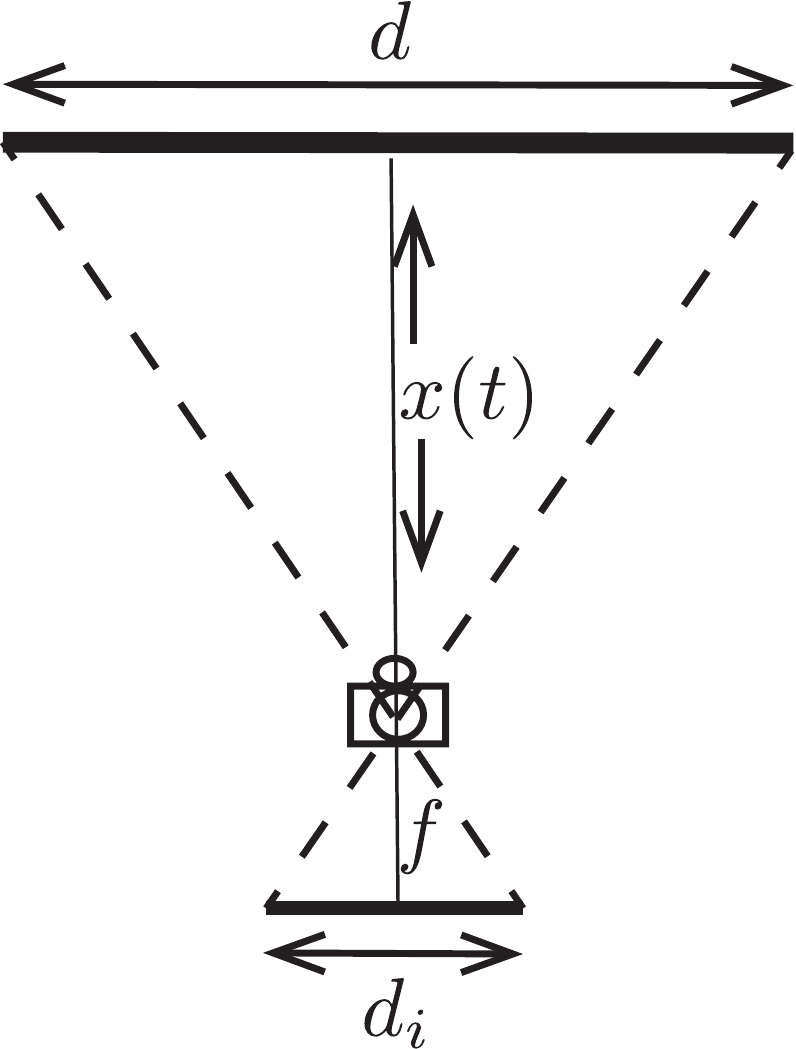}\end{center}
\caption{A thin object of diameter $d$ is approached at a constant velocity $\dot x(t)=v$. We assume the obstacle registers on the image plane of a pinhole camera, and the length of the image is $d_i$.} 
\label{fig:jb:OpticalFlow} 
\end{figure}  
If the obstacle is observed using a pinhole camera as depicted, it creates an image on the focal plane of the camera.  The variables here depicted are $d=$ {\em diameter of the object}, $d_i=$ {\em diameter of the image on the image plane}, $x=$ {\em distance from the camera to the object}, and $f=$ the focal length of the camera.  These variables are related by the {\em fundamental constitutive equation}
\begin{equation}
\frac{d}{x}=\frac{d_i}{f}.
\label{eq:jb:basic}
\end{equation}
Suppose the camera moves toward the object at constant speed, $\dot x(t)=-v$ with $v>0$.  Denote $x(0)$ by $x_0$.  As $x(t)$ decreases, $d_i=d_i(t)$ increases.  The quantities $f,d$ remain fixed.  Using the constitutive equation (\ref{eq:jb:basic}), it is possible to derive a differential equation governing the evolution of the image size $d_i$:
\begin{equation}
\dot d_i(t)=\frac{v}{d\cdot f}d_i(t)^2.
\label{eq:jb:OpticalFlow}
\end{equation}
When $v$ is held to any constant positive value, it is not surprising that there is a finite escape time for (\ref{eq:jb:basic}).  The easily derived solution 
\[
d_i(t)=\frac{d\cdot f}{x_0- v\cdot t}
\]
shows this to be $\tau=x_0/v$.  Remarkably, if it is assumed that none of the quantities $d$, $x_0$, or $v$ is known, 
it is still possible to determine $\tau$.  Rewriting (\ref{eq:jb:basic}) as
\[
d\cdot f = d_i(t)\cdot x(t), 
\]
and differentiating both sides, we have $\dot d_i(t)\,x(t)-v\,d_i(t)=0$.   Under the assumption of constant closing velocity, this may be written as
\[
\frac{d_i}{\dot d_i}=\frac{x_0-v\,t}{v}.
\]
This is zero when $t=x_0/v$ (the time of contact), and at $t=0$, we see that $d_i/\dot d_i = x_0/v =\tau$ is the time remaining until contact. 

While this analysis has been done in terms of the observed changes in an image of an ideally situated idealized object, it in fact applies more generally to optical flow associated with any feature point in the field of view.  The general conclusion is that if $d_i(t)$ is the location of an image feature in the image plane, $\tau=d_i(t)/\dot d_i(t)$ is the time remaining until the camera is directly abeam of the actual feature.  This will be discussed in greater detail in Section \ref{sec:jb:motion}.

The time-to-contact parameter that we have called $\tau$ has been the focus of research in the psychology of human and animal perception.  Kaiser and Mowafy (\cite{Kaiser}) have noted that concepts closely related to time-to-contact are relevant in a variety of settings in which motion perception is important.  In particular, they have reported experiments in which humans are required to use optical flow information from screen images to estimate the time remaining before a simulated object is passed.  The work indicates that subjects have greater difficulty estimating time-to-pass for objects that are significantly off the axis of approach.  Lee and Reddish (\cite{Lee-Reddish}) present evidence that $\tau$ plays a role in the way that gannets time their dives into the ocean in attempts to catch fish, and Wang and Frost (\cite{Wang-Frost}) have reported finding a part of the brains of pigeons that respond selectively to approaching objects in a way that is consistent with $\tau$ being the key perceptual variable.  Like the finding of Kaiser and Mofawy, the response in the pigeon brain was markedly reduced when the object deviated 5\% or more from a collision course.

The contribution of the present paper is to develop an understanding of $\tau$ in cases where the optical sensor is moving at a nonconstant velocity and to use this understanding to design motion control algorithms based on sets of $\tau$'s associated with selected features in the sensor image. We are not aware of previous attempts to do this, and it is hoped that the kinds of designs under development will be useful in creating animal-like controlled motions of UAV's.  In Section 3 below, we introduce the concept of {\em time-to-transit} with respect to a feature point, which we denote by $\tau(t)$, the time-dependency indicating that $\tau$ may vary according to how the optical sensor accelerates or follows a curved path with respect to the feature point.

The organization of the paper is as follows.  Section II proposes a quantized steering model of a Dubins-like model of planar nonholonomic motion.  Inspired by work by Karaman and Frazzoli (\cite{frazzoli}), the controlled motions of quantized Dubins vehicles through what we term Markovian obstacle fields are analyzed.  It is shown that the likelihood of a controlled motion successfully transiting an obstacle field without collision depends on both the size and density of obstacles in the obstacle field and the steering authority (effectively the turning radius) that can be applied by the vehicle's controller.  For a wide range of obstacle densities, it is shown that there are critical levels of steering authority, slightly below which it is almost impossible to transit an obstacle field and slightly above which it is almost certain the there will be a realizable collision-free path.  In Section III, we consider motion control based on optical sensing and introduce the concept of {\em time-to-transit} a feature point.  Control laws based on sensing of proximity and heading are discussed along with a control law based on times to transit certain key feature points.  Section IV describes an implementation of optical flow algorithm on an actual quadrotor UAV using a GoPro camera.  Conclusions and plans for further research are given in Section V.

\section{A Markovian Obstacle Field}
\label{sec:jb:Markovian}

Consider an infinite line of one-dimensional obstacle slats as depicted in Fig.\ \ref{fig:jb:ObstacleField}.  The lengths and positions of each slat are random variables whose lengths and placement a modeled by a Markov law as follows.  Let $s$ denote position along the infinite line.  The line is partitioned into a countable number of adjacent open and closed segments.  We denote the defining partition as $\{\dots<s_{-1}<s_0<s_1<\dots\}$ which is alternately composed of obstacle slats (closed segments of the form $[s_{2k},s_{sk+1}]$) and open spaces (open segments ($s_{2k-1},s_{2k}$)).  The position $s$ lies in an open space with probability $p_1(s)$ and on an obstacle slat with probability $p(_2(s)$.  Note that $p_1(s)+p_2(s)\equiv 1$.

A point $s$ may or may not lie on an obstacle slat, and in either case, as we consider points to the right of $s$, a slat boundary will be encountered sooner or later.  As the value of $s$ increases the dependence of our probabilities on the spatial variable is
\[
\begin{array}{ccc}
p_1(s+\Delta s) & = &(1-\beta\,\Delta s)\,p_1(s)+\alpha\,\Delta s\, p_2(s) +  o(\Delta s)\\
p_2(s+\Delta s) & = &\beta\,\Delta s\,p_1(s) + (1-\alpha\,\Delta s)\,p_2(s)+  o(\Delta s),
\end{array}
\]
where as $s$ increases, the probability of coming to the left edge of a slat in the next $\Delta s$ inches, given that the point $s$ lies in open space, is $\beta\,\Delta s + o(\Delta s)$, while the probability of coming to the right edge of a slat, given that $s$ lies on a slat, is  $\alpha\,\Delta s + o(\Delta s)$.  In the limit, as $\Delta s\to 0$, these equations may be rendered
\begin{equation}
\begin{array}{ccc}
dp_1/ds & = & -\beta\, p_1(s) + \alpha\, p_2(s),\\
dp_2/ds & = &  \beta\, p_1(s) - \alpha\, p_2(s).
\end{array}
\label{eq:jb:Markov}
\end{equation}

We assume that the distribution of slats is spatially stationary, and it is immediate that the stationary distribution is $p_1=\alpha/(\alpha+\beta),\, p_2=\beta/(\alpha+\beta)$.  Thus, for any point $s$, the probability of $s$ lying on a slat is $p_2$, and the probability that it lies in open space is $p_1$.

We can describe that statistical characteristics of the sizes of obstacles and open spaces as follows.  Suppose the point $s=0$ lies in open space between obstacles.

Set up in this way, the distribution governing the occurrence of an obstacle edge as the variable $s$ increases in an interval of open space is exponential with parameter $\beta$.  This means that the mean width of the open space is $1/\beta$, while the variance is $1/\beta^2$.  Similarly, the right hand obstacle edge is exponentially distributed with parameter $\alpha$, and the mean obstacle width and variance are $1/\alpha$ and $1/\alpha^2$ respectively.   Assuming the open space is on average wider than the obstacles ($1/\beta > 1/\alpha$), the variance in the size of the open spaces is larger than the variance in obstacle widths as well.

\begin{figure}[ht] 
\begin{center}\includegraphics[width=3.0in]{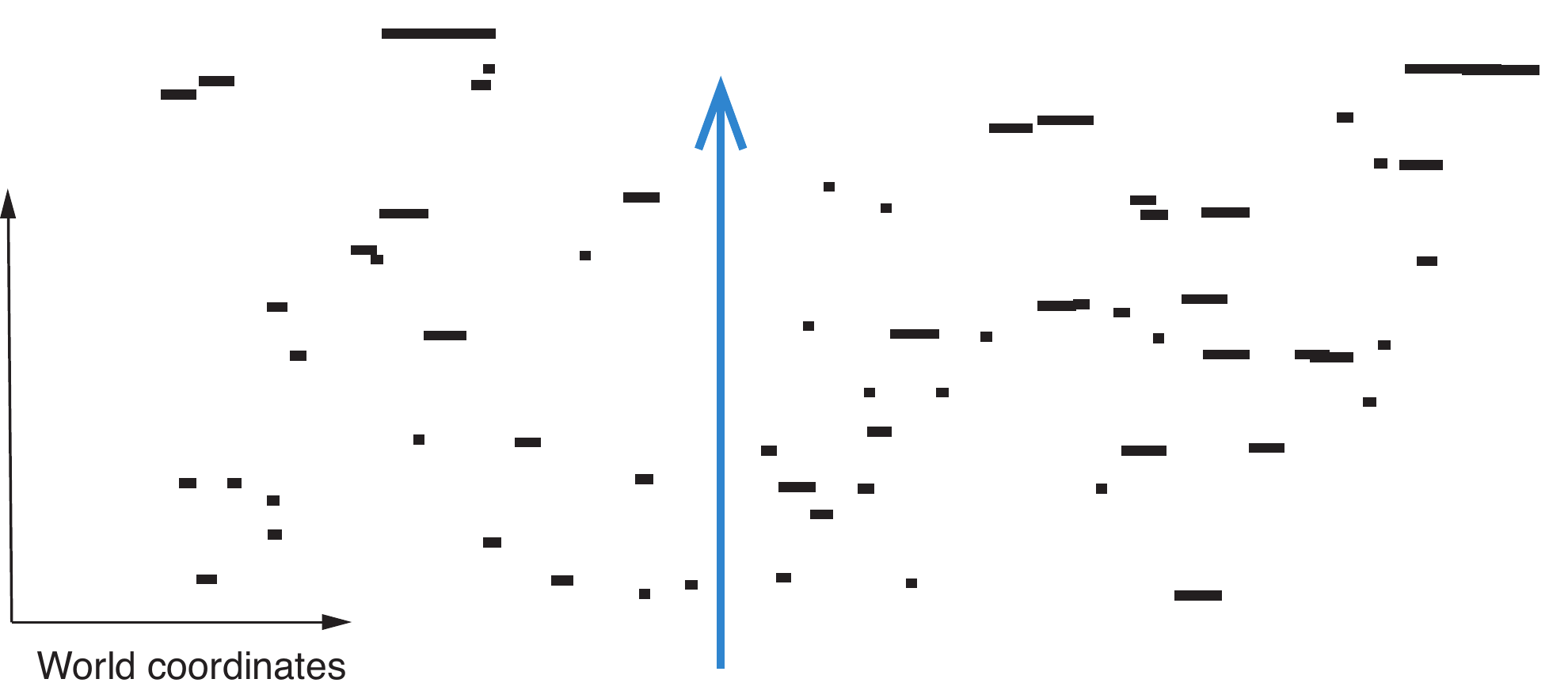}\end{center}
\caption{A Markovian obstacle field.  Obstacle heights are one pixerl while widths follow an exponential distribution with ``rate'' parameter $\alpha=1.0$.  The between-the-obstacle open spaces follow exponential distributions with rate parameters $\beta=0.1,\gamma=0.1$ in the $x$ and $y$ directions respectively.  With these parameters, there is clearly a reasonable likelihood that straight flight through the obstacle field as illustrated will not be impeded.} 
\label{fig:jb:ObstacleField} 
\end{figure}  

This characterization of one-dimensional arrays of obstacle slats provides a basic building block for constructing an idealized model of a two-dimensional obstacle field.  In Fig.\ \ref{fig:jb:ObstacleField}, rows of one-dimensional obstacle arrays are stacked vertically.  To represent exponential randomness in the $y$-direction, we assume that each obstacle is displaced from a mean row position by an amount that is exponentially distributed with parameter $\gamma$.  In Figure \ref{fig:jb:ObstacleField}, the average obstacle width has parameter $1/\alpha$ where $\alpha=1.0$, while the average open space in both the horizontal and vertical directions have length parameters $1/\beta=1/\gamma=10$.

Having thus constructed a two-dimensional Markovian obstacle field, it is easy to compute the likelihood of a point vehicle being able to transit its depth (in the $y$-direction) without colliding with an obstacle.  A straight collision-free path is depicted by the blue arrow.  Assuming that the obstacles may be grouped into horizontal rows in each of which the obstacle widths and separations are given by the stationary versions of (\ref{eq:jb:Markov}) (i.e.\ $p_1=\alpha/(\alpha+\beta),\, p_2=\beta/(\alpha+\beta)$), we find that the probability of transiting exactly $n$ rows before encountering an obstacle is 
\[
P_n=p_1^{n}p_2.
\]
From this, one can calculate the probability of unobstructed motion through at least $n$ rows of obstacles:
\[
Q_n=1-\sum_{k=0}^{n-1}P_k=p_1^n.
\]
Finally, the mean length of a collision-free path entering the obstacle field from a randomly chosen point along the horizontal axis is
\[
S=\sum_{k=0}^{\infty}kP_k=p_1/p_2,
\]
with the corresponding variance being 
\[
\sum_{k=0}^{\infty}\left(k-\frac{p_1}{p_2}\right)^2P_k=\frac{p_1}{p_2}\left(1+\frac{p_1}{p_2}\right).
\]
That the standard deviation is always larger than the mean indicates that uncontrolled motion through even a sparse obstacle field is risky.  
For the parameter values of Figure \ref{fig:jb:ObstacleField}, the mean unobstructed path is ten length units long, while the variance is 110 ($\sim$ standard deviation 10.49).


Continuing the abstraction of the obstacle field as rows of one-dimensional slat obstacles, we come to the question of the probability of constructing a collision-free path by a point vehicle through an obstacle field that contains a specified number of rows.  It is assumed that the point vehicle  in question executes Dubins-type motions---the vehicle's orientation is specified by a parameter $\theta$, its forward velocity $(\dot x,\dot y)$ has unit magnitude and is constrained to satisfy $\dot y\cos\theta-\dot x\sin\theta=0$, and its turning radius cannot be less than some minimum $R_{cr}$.  Under these assumptions, the point vehicle may not be able to avoid colliding with a large obstacle lying directly in its path.

To connect the kinematics of a steerable planar point vehicle with the above probabilistic analysis, we introduce a quantized idealization of Dubins vehicle motion.  Specifically, we assume the vehicle always moves along a straight line path but can instantaneously change direction by an amount $\Delta\theta$ such that $0\le |\Delta\theta |\le\theta_{cr}$.

In the {\em quantized Dubins vehicle} kinematics, a circular arc segment is replaced by an instantaneous stright line path that has the same beginning and ending points as the arc.  This straight line segment is the chord of the arc between the given endpoints.  The quantized Dubins vehicle is steered though the obstacle field by means of the following control protocol.  Note that all path segments are straight lines.  Note also that the protocol admits the possibility of collision with an obstacle.  The probability of such a collision is given by the proposition that follows.

\begin{figure}[ht] 
\begin{center}\includegraphics[width=3.5in]{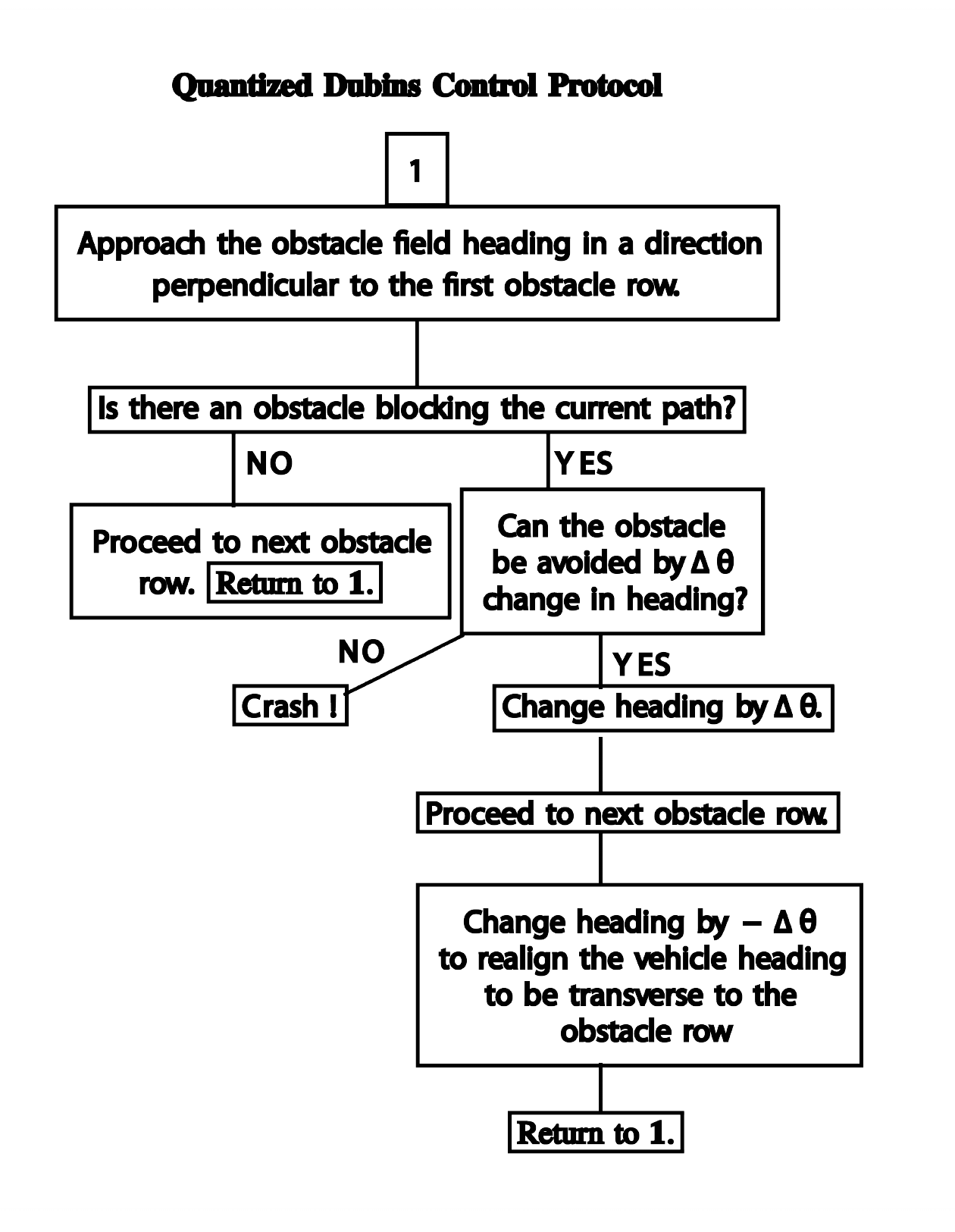}\end{center}
\vspace{-0.2in}
\label{fig:jb:ControlProtocol} 
\end{figure}

\begin{proposition}
Consider an $n$-row obstacle field in which the obstacles in each row lie along a line with the $k$-th row being displaced from the $k-1$-st row by an amount $1/\gamma$.  
The probability that there exists a collision free path through such an $n$-row obstacle field is
\begin{equation}
\left(p_1+p_2(1-e^{-(\alpha /\gamma) \tan\theta_{cr}})\right)^n.
\label{eq:jb:CollisionFree}
\end{equation}
\label{prop:jb:CollisionFree}
\end{proposition}
\begin{figure}[ht] 
\begin{center}\includegraphics[width=3.0in]{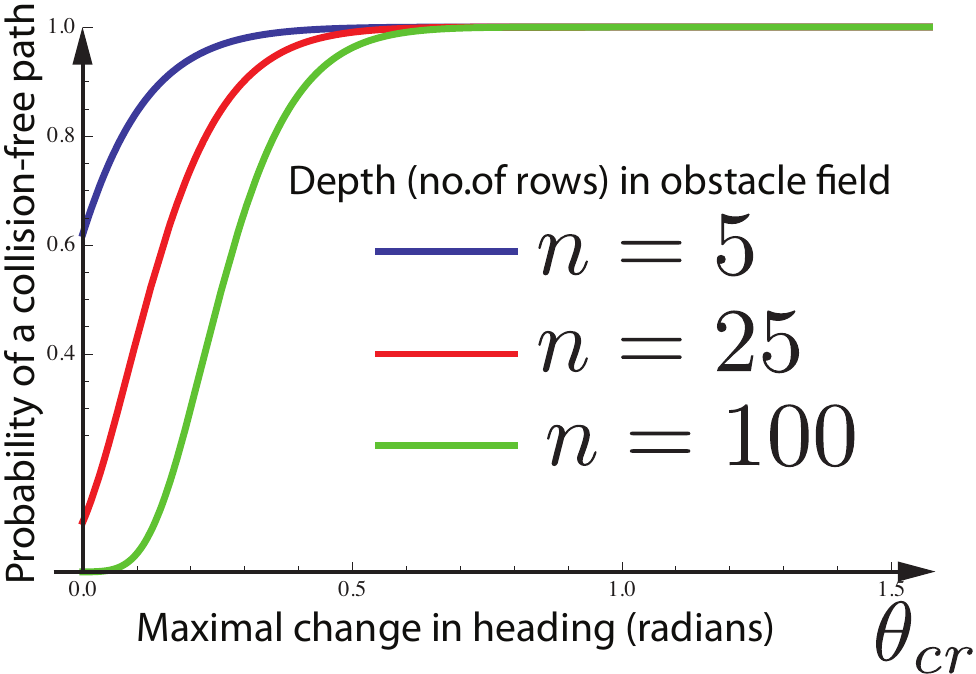}\end{center}
\caption{The probabilities of collision-free paths through $n$ evenly spaced rows of obstacles as functions of Markovian parameters $\alpha,\beta$, depth parameter $1/\gamma$, and critical steering angle $\theta_{cr}$.} 
\label{fig:jb:CollisionProbability} 
\end{figure}  
\begin{remark} \rm
For the parameter values of Fig.\ \ref{fig:jb:ObstacleField} the dependence of this probability on $\theta_{cr}$ is shown in the next figure.  Even modest steering capability dramatically increases the likelihood of existence of a collision-free path.  Moreover, for each number $n$ of obstacle rows, there is a value of steering angle $\theta_{cr}^*$ above which it is virtually certain that a collision free path exists.
\end{remark}

\begin{remark} \rm
Although the details of the curve shapes differ as the parameters $\alpha$ and $\beta$ vary, the form of the dependence of (\ref{eq:jb:CollisionFree}) on $\theta_{cr}$ does not change qualitatively.  Thus, even for dense arrays of thick obstacles, collision-free paths are possible---provided there is sufficient steering authority.
\end{remark}

\begin{remark} \rm
The proposition will be proved under the assumption that the obstacles lie on horizontal lines as opposed to having exponentially distributed probabilistic separation in the "$y$-axis" direction, which is the case depicted in Fig.\ \ref{fig:jb:ObstacleField}.  This uniformity assumption can be relaxed, and we expect that the probability of existence of obstacle avoiding paths will have a qualitatively similar expression.
\end{remark}

\begin{proof} 
(Of Proposition \ref{prop:jb:CollisionFree})  There are exactly two ways in which the point vehicle can avoid colliding as it approaches a row of obstacles.  It may be the case that it approaches the obstacle row at a point lying in open space between obstacles.  The probability of that is $p_1$.  It may also be headed for a collision with an obstacle (the probability of which is $p_2$) and yet steer left or right by an amount $\pm\Delta\theta$ where $\Delta\theta\le\theta_{cr}$ using the steering protocol described above.  The probability that there is enough steering authority to get past the left or right edge of the obstacle is $1-e^{-(\alpha /\gamma) \tan\theta_{cr}}$.  Since being able to steer past an obstacle by changing heading by an amount less than $\theta_{cr}$ is independent of whether an obstacle is encountered, the joint probability of encountering and obstacle AND being able to avoid it is $p_2(1-e^{-(\alpha /\gamma) \tan\theta_{cr}})$.  Thus, the probability of being able to avoid a collision in any row is 
\[
p_1+p_2(1-e^{-(\alpha /\gamma) \tan\theta_{cr}}).
\]
Since the rows have i.i.d. probabilities of not colliding, the probability of a collision-free path through $n$ rows of obstacles is given by (\ref{prop:jb:CollisionFree}) as claimed.
\end{proof}

\section{Motion control using optical sensor feedback}
\label{sec:jb:motion}

In order to prescribe a steering algorithm based on optical flow servoing, we revisit  the simple model of {\em time-to-contact} introduced in Section \ref{sec:jb:Intro}.  The goal will be to develop motion control laws that use images and sensed optical flow to automatically avoid obstacles.  For simplicity, only planar motion is considered.  While there is no essential difficulty in extending the treatment to the case of motion in three dimensions, the restriction to 2-D does not appear to be overly restrictive in trying to understand animal flight through forests since bats and birds tend to pick flight levels above or immediately below canopies of branches and foliage where the only obstacles are the trunks of trees.  (See e.g.\ \cite{Caceres}.)  We extend the analysis of Section \ref{sec:jb:Intro} by considering a more general class of motions of a mobile optical sensor as depicted in Figure \ref{fig:jb:OpticalFlowVersion}.
\begin{figure}[ht] 
\begin{center}\includegraphics[width=2.7in]{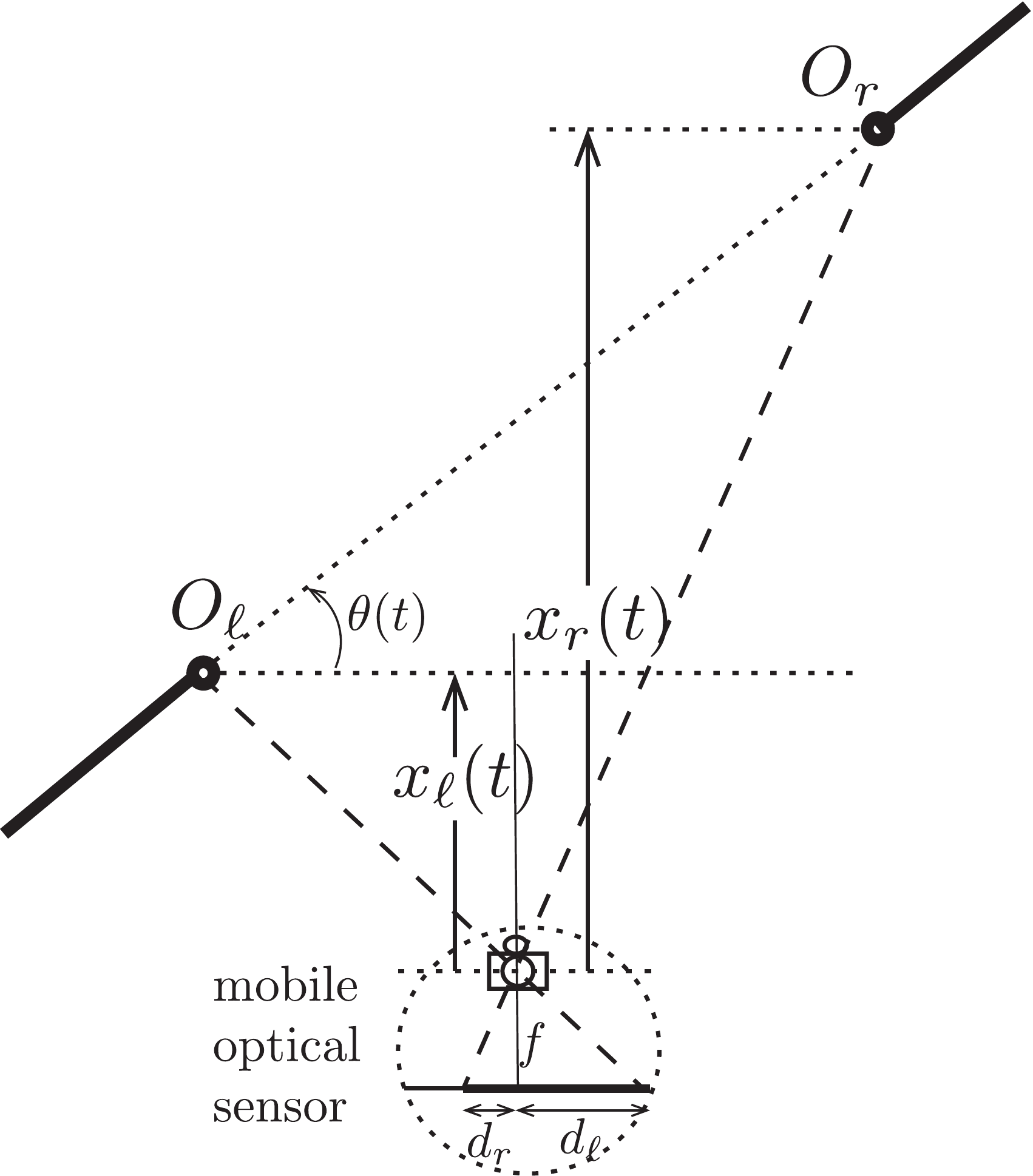}\end{center}
\caption{Thin objects and the open space between are approached at a constant velocity $\dot x(t)=v\cdot t$.} 
\label{fig:jb:OpticalFlowVersion} 
\end{figure}  

{\bf Coordinate frames.}  We assign a body-frame to the mobile optical sensor in a way that is consistent with standard kinematic models of planar nonholonomic vehicles.  Thus, the body-frame $x$-axis is aligned with the direction of motion, and consequently, the image is projected onto the body-frame $y$-axis.  We continue to consider Markovian obstacle fields as described in Section \ref{sec:jb:Markovian}.  Thus we consider linear arrays of obstacle slats (line segments) between which it is desired to have the mobile sensor move.  A frame of world coordinates is assigned such that the obstacle slats (line segments) are aligned with the world-frame $x$-axis, and the direction in which the sensor hopes to move is aligned with the world-frame $y$-axis.  

The mobile sensor coordinates at time $t$ are $\left(x(t),y(t),\theta(t)\right)$, giving its position and orientation with respect to the world frame.  The motion is described by a kinematic model:
\begin{equation}
\left(\begin{array}{c}
\dot x \\
\dot y \\
\dot\theta\end{array}\right) = \left(\begin{array}{l}
v\cos\theta \\
v\sin\theta \\
\omega\end{array}\right),
\label{eq:jb:BasicVehicle}
\end{equation}
where $v$ is its forward speed (in the direction of the sensor's body-frame $x$-axis, and $\omega$ is the turning rate. 

If $(x_r,y_r)$ is the coordinate pair that locates the obstacle point $O_r$ in world coordinates, the corresponding point on the image (a point along the body-frame $y$-axis) is
\begin{equation}
d_r(t)=\frac{-\sin\theta(t)(x_r-x(t))+\cos\theta(t)(y_r-y(t))}{1-\cos\theta(t)(x_r-x(t))-\sin\theta(t)(y_r-y(t))}.
\label{eq:jb:RightDistance}
\end{equation}
It is assumed that the optical sensor focal lenth $f=1$, 
here $(x(t),y(t),\theta(t))$ are the position and orientation of the sensor frame with respect to world coordinates at time $t$.  A similar relationship holds for $O_{\ell}$ and $d_{\ell}(t)$.  If the mobile sensor proceeds along the straight line that is shown in Fig.\ \ref{fig:jb:OpticalFlowVersion}, it will cross a line that passes through $O_{\ell}$ and is perpendicular to the line of travel before it crosses a similar line perpendicular to its path and passing through $O_r$.  Note that this line is rotated by an amount $-\theta(t)$ with respect to the line of obstacle slats.  As depicted, if the optical sensor continues along its straight line path, it will pass between the obstacles that are shown.  Following the reasoning of Section \ref{sec:jb:Intro}, the time to crossing the line through $O_{\ell}$ will be $d_{\ell}/\dot d_{\ell}$, provided a straight line motion at constant velocity (say $v$) is carried out. This assumption of straight-line, constant speed motions is crucial in order to interpret $d_{\ell}/\dot d_{\ell}$ as the {\em time-to-transit} the line through $O_{\ell}$.  To pursue this assumption in computing and interpreting $\dot d_{\ell}$ or $\dot d_r$, we treat $\theta$ as a constant: $\theta(t)\equiv\theta$.  A straightforward calculation then shows that (\ref{eq:jb:RightDistance}) may be more simply stated in terms of the world frame initial position $(x(0),y(0))$ and the constant speed $v$ as
\[
\begin{array}{ccl}
d_r(t)&=&\frac{\cos\theta(y_r-y(0)-v\sin\theta)-\sin\theta(x_r-x(0)-v\cos\theta)}{1-cos\theta (x_r-x(0)-v\cos\theta) \ \sin\theta (y_0-y(0)-v\sin\theta)}\\[0.1in]
&=&\frac{\cos\theta(y_r-y(0))-\sin\theta(x_r-x(0))}{1-cos\theta (x_r-x(0)) - \sin\theta (y_0-y(0))+vt}.
\end{array}
\]
Taking the derivative with respect to $t$, the {\em time-to-transit} the perpendicular line through $O_r$ given a constant straight line velocity $v$ is 
\[
d_r(t)/\dot d_r(t)=\frac{\cos\theta (x_r-x(0))+\sin\theta(y_r-y(0))-1}{v}-t.
\]
The interpretation of this formula is that if constant velocity motion is carried out in the direction $(\cos\theta,\sin\theta)$ given the starting point $(x(0),y(0))$, then the sensor will transit the line perpendicular to its path and passing through the point $(x_r,y_r)$ when $t=(\cos\theta (x_r-x(0))+\sin\theta(y_r-y(0))-1)/v$.  In other words, the time-to-transit from the starting point $(x(0),y(0))$ assuming a straight line velocity $(v\cos\theta,v\sin\theta)$ is 
\[
\tau_r=d_r(0)/\dot d_r(0)=(\cos\theta (x_r-x(0))+\sin\theta(y_r-y(0))-1)/v.
\]


These derivations illustrate the relationships between $\tau_r$ and the vehicle kinematic variables.  It is important to re-emphasize, however, that $\tau_r$ is determined solely from quantities measured by the image sensor, and no knowledge of $x,y,\theta$, or $v$ is assumed.  Similar considerations apply to the left feature point $O_{\ell}$ and $\tau_{\ell}=d_{\ell}(0)/\dot d_{\ell}(0)$.

It is possible to use inferred time-to-transit quantities to generate steering signals for motion control of the mobile point sensor as follows.  In the formula for $\tau_r$, the starting point $(x(0),y(0))$ may be chosen arbitrarily.  In particular, it could be any point along a vehicle trajectory, and hence for any trajectory $(x(t),y(t),\theta(t))$, we write
\begin{equation}
\begin{array}{ccl}
\tau_r(t)&=&\bigl(\cos[\theta(t)] (x_r-x(t))+\sin[\theta(t)](y_r-y(t))-1\bigr)/v,
\label{eq:jb:TimeToTransit}
\end{array}
\end{equation}
which has the interpretation that if the mobile sensor were to proceed at constant speed $v$ in the direction $(\cos\theta(t),\sin\theta(t))$ along the path segment
\[
\left(\begin{array}{c}
\hat x(s)\\
\hat y(s)
\end{array}\right)=
\left(\begin{array}{c}
x(t)+vs\cos\theta(t)\\
y(t)+vs\sin\theta(t)
\end{array}\right),\ \ \ 0\le s\le \tau_r(t)
\]
for $\tau_r(t)$ units of time, it would transit the line perpendicular to this straight line path and passing through $(x_r,y_r)$.  We note that thinking of the transit times $\tau_r(t)$ and $\tau_{\ell}(t)$ as variable quantities that depend on the vehicle kinematics in this way allows useful inferences (e.g.\ is the vehicle slowing or accelerating) about motion relative to the associated feature points.

\subsection{Steering control using image data feedback}

As noted in \cite{Baillieul03},\cite{Baillieul04}, kinematic control laws using sensor measurements referenced to the body frame must be understood in light of Brockett's nonstabilizibility theorem. More specifically, no configuration $\left(x_1,y_1,\theta_1\right)$ can be asymptotically stabilized by a kinematic control law in which the velocity variables $v,\omega$ in (\ref{eq:jb:BasicVehicle}) depend continuously on sensor measurements.  This is well illustrated by a feedback law based on distance and bearing measurements of the type provided by LIDAR or binocular vision.  Suppose such a sensor accurately measures the vehicle's distance $\rho$ and bearing $\phi$ with respect to a feature point---say $O_r=(x_r,y_r)$.  In \cite{Baillieul04}, it was shown that for selected values of a gain parameter $\lambda$ the feedback control law 
\begin{equation}
v(\rho)=\lambda(\rho-\phi),\ \ {\rm and}\ \ \omega(\rho,\phi)=\rho\sin\phi-d
\label{eq:jb:perp.rho.phi}
\end{equation}
will always steer the system (\ref{eq:jb:BasicVehicle}) to a point on a circle of prescribed radius $d$ centered at a goal point.  More specifically, we have the following theorem whose proof is given in \cite{Baillieul04}.

\begin{theorem} (\cite{Baillieul04})
Consider the controlled system (\ref{eq:jb:BasicVehicle}) with initial conditions
$(x_0,y_0,\theta_0)$, and consider a goal point with coordinates $(x_r,y_r)$.
For the control law (\ref{eq:jb:perp.rho.phi}), suppose $\lambda$ is in the interval $0<\lambda<d$. 
If the initial configuration of the vehicle body frame $(x_0,y_0,\theta_0)$ is
not directed along the unstable branch of $\rho\sin\phi-d=0$, then the
motion of (\ref{eq:jb:BasicVehicle}) tends asymptotically to $(\rho,\phi)
=(d,\pi/2)$.  The corresponding $(x,y,\theta)$ trajectory generated by
(\ref{eq:jb:BasicVehicle}) tends asymptotically to a point on the circle
$$(x-x_r)^2+(y-y_r)^2 = d^2$$
with the body-frame $x$-axis [in the direction $(\cos\theta,\sin\theta)$]
tangent to the circle, pointed in the counterclockwise direction.  In a completely symmetric fashion,
 the feedback law
$v=\lambda(\rho - d)$, $\omega=\rho\sin\phi + d$ produces 
trajectories tangentially approaching the circle of radius $d$ about 
the goal point in the clockwise direction of tangency.
\label{thm:jb:Baillieul04}
\end{theorem}

\begin{remark} \rm
{\em Servoing on feature perceptions:}  One of the great challenges in using body-referenced sensing is illustrated by the {\em singular variety} $\rho\sin\phi-d=0$ and its apparent movement in the world frame as the vehicle moves.  The singular variety is a manifestation of Brockett's nonstabilizibility theorem, and it specifies a set of configurations from which the goal points cannot be reached by the control law (\ref{eq:jb:perp.rho.phi}).  In $(x,y)$ coordinates, the singular variety can be constructed as follows.  Draw a circle of radius $d$ around the goal point $(x_r,y_r)$.  Next draw a line from the current vehicle position $(x,y)$ to the goal point.  Draw a line that is is tangent to the circle and that makes an angle $\pi-\phi$ with the previous line as depicted in Fig.\ \ref{fig:jb:Asympt-Set}.  This is the stable branch of the singular variety.  The semi-infinite red line making an angle $\phi$ with the line to the goal point is the unstable branch.  This unstable branch characterizes the initial configurations from which the control law fails to asymptotically approach tangency to the circle.  Singular sets of this type are characteristic of systems to which Brockett's nonstabilizability theorem applies.

\end{remark}
\begin{figure}[ht] 
\begin{center}\includegraphics[width=2.7in]{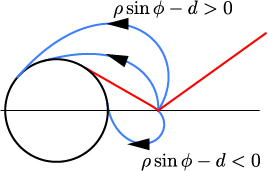}\end{center}
\caption{The singular set (red line segments) associated with the distance/bearing feedback law (\ref{eq:jb:perp.rho.phi}).  Given the goal point at the center of the circle and the current vehicle at the intersection of the two lines, it is the point locus of $0=\rho\sin\phi-d$.} 
\label{fig:jb:Asympt-Set} 
\end{figure}  

As the vehicle motion proceeds, the world-frame locus of the singular variety changes---with the vehicle's position always being at the intersection of the two branches of the singular locus.  The singular locus is characteristic of the way that features evolve when rendered in terms of world coordinates.  Similar feature dynamics are seen in the {\em lines of transit} associated with a feature point and the {\em time-to-transit} parameter $\tau$ introduced above.  As the vehicle moves, these lines of transit rotate about the feature point.  This apparent motion plays a role in the design and analysis of a time-to-transit feedback law that enables passage between the feature points $O_{\ell}$ and $O_r$.

In \cite{Baillieul04}, the proof of Theorem \ref{thm:jb:Baillieul04} is carried out in two steps, the first of which specifies a set in terms of the body-referenced variables $\rho$ and $\phi$ that is invariant under the controlled motions.  The second step is to show that all motions starting in this invariant set tend asymptotically to the specified limits.  In the following, a control law and its invariant set are specified in terms of our image parameters.

\begin{theorem}
Consider the mobile sensor depicted in Fig.\ \ref{fig:jb:OpticalFlowVersion} and having motion kinematics (\ref{eq:jb:BasicVehicle}).  Let $\epsilon>0$ be a small positive constant.  The image-referenced configuration set
\begin{equation}
d_{\ell}\le-\epsilon\ {\rm and}\ \ d_r\ge\epsilon
\label{eq:jb:InvariantSet}
\end{equation}
is invariant under the feedback control law 
\[
\begin{array}{ccl}
v(t)&=&\min(1,\tau_{\ell}+\tau_r)\\
\omega(t)&=&\left\{\begin{array}{ccl}
0&{\rm if}&d_r=\epsilon\ {\rm or}\ d_{\ell}=-\epsilon\\
\tau_r-\tau_{\ell}&{\rm if}& d_r>\epsilon\ {\rm and} \ d_{\ell}<-\epsilon.
\end{array}\right.
\end{array}
\]
\end{theorem}
\begin{proof} [Sketch] 
Under the assumed model of the obstacle field, we have $y_r=y_{\ell}$.  I.e.\ the left obstacle and the right obstacle, between which we wish to pass, are the same distance from the $x$-axis.  We assume that the desired motion is in the direction of the positive $y$-axis, and hence throughout the time over which the control law operates, we have the vehicle's $y$ coordinate $y < y_r$.

We consider two cases:  First, suppose that the vehicle is located at a point $(x,y)$ with $x_{\ell}\le x\le x_r$ and $y<y_r$.  If this is the case and (\ref{eq:jb:InvariantSet}) holds, the heading of the vehicle is such that in the absence of steering (i.e.\ $\omega=0$) it would approach the line segment connecting $(x_{\ell},y_{\ell})$ and $(x_r,y_r)$.  It can be shown that the set $x_{\ell}<x<x_r$, $y<y_r$, and (\ref{eq:jb:InvariantSet}) is invariant, and that for some $T>0$ and all $t>T$, $d_{\ell}<-\epsilon$, $d_r>\epsilon$.  Thus, the vehicle will be steered by $\omega(t)=\tau_r(t)-\tau_{\ell}(t)$, and it will turn so as to equalize these transit times, and thus it asymptotically approaches the line segment between $(x_{\ell},y_{\ell})$ and $(x_r,y_r)$ in an asymptotically perpendicular direction.

In the second case, the vehicle position $(x,y)$ is assumed to lie outside the strip $x_{\ell}\le x\le x_r$.  We assume without loss of generality that $x>x_r$.  (The case $x<x_{\ell}$ is treated symmetrically.)  If the vehicle is such that $d_r>\epsilon$, then because $\tau_r-\tau_{\ell}<0$, it will turn toward the point $(x_r,y_r)$.  It may enter the strip $x_{\ell}\le x\le x_r$, in which case the motion beyond that point is as discussed in case one above.  If it does not enter the strip, the value of $d_r$ will continue to decrease.  Once $d_r=\epsilon$, it will remain at that value until the vehicle has nearly reached the line segment joining $(x_{\ell},y_{\ell})$ and $(x_r,y_r)$.  At that point $\tau_{\ell}\sim \tau_r$, and the vehicle approaches the line segment in an asymptotically perpendicular direction, and it passes close to the point $(x_r,y_r)$.
\end{proof}
Space does not allow a more complete proof.  This will be given elsewhere.

\section{Optical Flow Implementation}

The research described above is supported by an application platform which will be useful in understanding issues in our study of animal-inspired flight dynamics.  We have developed a custom UAV based on the popular quadcopter airframe, equipped with motion sensors, an onboard camera, and a Gumstix Fire SBC (Single Board Computer). Our aim is to demonstrate that the real-time, robust measurement necessary for navigation as discussed earlier in the paper is possible with these well established technologies together with appropriately selected algorithms for real-time optical flow processing.

Modern optical flow algorithms can be roughly grouped into two families: \emph{sparse flow} and  \emph{dense flow}. In sparse optical flow, only certain points are tracked, and relatively simple models such as (\ref{eq:jb:OpticalFlow}) describe the motion of feature points on the image plane.
In dense optical flow the entire warping transformation is described by the well known intensity flow equation
\begin{equation}
\frac{\partial{I}}{\partial{x}} v_x + \frac{\partial{I}}{\partial{y}} v_y  + \frac{\partial{I}}{\partial{t}} = 0.
\label{eq:optical_flow}
\end{equation}

Sparse, i.e. local, methods are both more robust under noise and simpler to implement in real-time applications. Like all finite approximations to continuum problems, sparse methods can be criticized for not adequately allowing reconstruction of the entire flow field (\cite{Bar-Fle.1994}), but for feedback control applications their superiority seems clear. The basic sparse flow algorithm that has been tested is derived from OpenCV's Lucas-Kanade implementation.  Feature points are either found in high texture regions of the scene or are taken at mesh points, and these are tracked to give flow fields of adequate density to implement the types of control algorithms described in the previous section. Our lab is currently conducting indoor free flight tests together with selected tethered outdoor flight tests.

\subsection{Test flight examples}

\begin{figure}[ht] 
\begin{center}\includegraphics[width=3.5in]{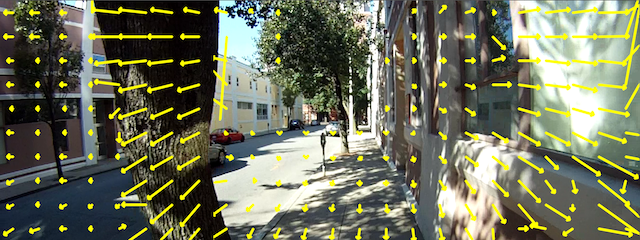}\end{center}
\caption{The kinetics can be clearly seen, as the tree stands out compared to its immediate surroundings. } 
\label{fig:ks:Cummington1} 
\end{figure}  

\begin{figure}[ht] 
\begin{center}\includegraphics[width=3.5in]{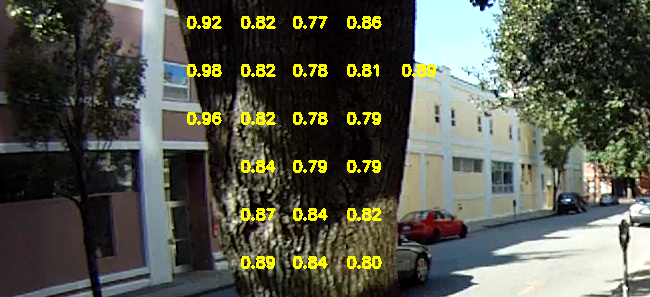}\end{center}
\caption{Showing time-to-contact in a close up of the tree.} 
\label{fig:ks:Cummington2} 
\end{figure}  

Figures \ref{fig:ks:Cummington1} and \ref{fig:ks:Cummington2} show views from an aerial vehicle flying down a narrow sidewalk, bracketed by trees on the left and a wall on the right. These images were used in calculating optical flow, without any image preprocessing (aside from the necessary conversion to greyscale). In addition, they were not chosen due to their lack of noise, but instead due to the interesting features they present.

In Fig.~\ref{fig:ks:Cummington1}, the proximity of the tree can clearly be inferred from the flow field's discontinuity. Objects in the background have minimal flow, but the tree itself has perhaps the highest flow in the scene. This demonstrates graphically the fact that it is not necessary to recognize obstacles in order to avoid colliding with them while using them as way-point beacons for navigating through cluttered environments..

Indeed, despite the 3D world, the UAV moves almost along a plane parallel to the ground, and so the approaches to calculating $\tau$ as given earlier in the work are perfectly applicable. Thus we calculate that the $\tau$ for the tree is 0.80 seconds, a calculation born out by the video data, which showed the traversal of a line perpendicular to the trees center in approx.~3/4 seconds after Fig.~\ref{fig:ks:Cummington1} was taken.

Figure \ref{fig:ks:Cummington2} shows a close up of the tree in Fig.~\ref{fig:ks:Cummington1}. The vectors have been replaced by a simple calculation of $\tau$. Note that the overall agreement is quite good, as the standard deviation is only $\approx$6ms, which on the scale of the UAV's movement is an order of magnitude faster than its dynamics.

\section{Conclusions and Further Work}

The paper has provided results on two interrelated but distinct aspects of the problem of autonomous vehicle flight control through fields of obstacles.  Results of Section II quantify the difficulty of constructing collision-free motions in terms of the size and density of the obstacles.  The mean length of a collision-free uncontrolled path has been shown to vary inversely with the average diameter of the obstacles, but it was shown that the variance of this variable grows faster than the square of the mean decreases---making uncontrolled motion very risky.  It was also shown that for what we termed a {\em quantized Dubins vehicle}, that the probability of being able to steer a vehicle along a collision-free path was sensitive to the level of steering authority, and for a field of obstacles of given average size and density there exists a narrow range of critical levels of steering authority below which it is almost impossible to move without colliding and above which it is virtually certain that a collision-free path can be followed.  In Sections I and III, several aspects of optical sensing for motion control were discussed.  Recalling the {\em time-to-contact} parameter $\tau$ from the perceptual psychology literature, we introduced the concept of {\em time-to-transit} and indicated how each feature point in a mobile sensor's field of view anchored a line of transit that rotated in world coordinates as the sensor moved.  Taking this motion into account, the parameter $\tau=\tau(t)$ was seen to vary along any flight path, and by selecting appropriate feature points---preferably kinetic boundaries between obstacles and free space, the corresponding variables $\tau(t)$ were shown to be usable as inputs to a control law.  An important feature of this type of control is that $\tau$ was shown to depend only on the sensor image, and it required no knowledge of the position and velocity of the sensor/vehicle.

Further research is needed to develop appropriate families of such control laws that can be used to steer through obstacle fields of complexity similar to the models of Section II.  It is only when complex fields are considered that an understanding of how such optical flow enabled control might degrade with respect to obstacle size and density.  Because this type of motion control appears to be employed in the natural world, it will be a great interest to see in the performance characteristics that we expect to develop will be observable in the flight behaviors of animals.


\begin{thebibliography}{99}
\bibitem{Baillieul03}
J.\ Baillieul and A.\ Suri, 2003.  ``Information Patterns and Hedging Brockett's Theorem in Controlling Vehicle Formations,''  {\em Proceedings of the 2003 IEEE Conference on Decision and Control}, Maui, Hawaii, December 9-12, TuMO2-6, pp.\  556-563,  DOI 10.1109/CDC.2003.1272622.
\bibitem{Baillieul04}
J.\ Baillieul, 2004.  ``The Geometry of Sensor Information Utilization in Nonlinear Feedback Control of Vehicle Formations,''  in {\em Cooperative Control: A Post-Workshop Volume 2003 Block Island Workshop on Cooperative Control}, June 10-11, 2003, V.\ Kumar, N.E.\ Leonard \& A.S.\ Morse,
Eds.\ Lecture Notes in Control and Information Sciences, Volume 309, Springer-Verlag, New York, ISBN:Ê3-540-22861-6.
\bibitem{Bar-Fle.1994}
J. L. Barron and D. J. Fleet and S.S. Beauchemin, 1994, ``Performance of optical flow techniques,"
{\em International Journal of Computer Vision.}
\bibitem{Bir-Pun.2008}
S. T. BirchÞeld and S. J. Pundlik, 2008, ``Joint Tracking of Features and Edges,"
{\em Proceedings of the IEEE Conference on Computer Vision
and Pattern Recognition.}
\bibitem{Bru-Wei.2005}
A. Bruhn and J. Weickert, 2005, ``Lucas/Kanade Meets Horn/Schunck: Combining Local and Global Optic Flow Methods,"
{\em International Journal of Computer Vision.}
\bibitem{Caceres}
M.C. Caceres and R.M.R. Barclay, 2000. ``Myotis septentrionalis,'' MAMMALIAN SPECIES No. 634, pp. 1-4, 3 Þgs., Published 12 May 2000 by the American Society of Mammalogists.
\bibitem{Huston}
Huston SJ, Krapp HG (2008) Visuomotor Transformation in the Fly Gaze Stabilization System. {\em PLoS Biol}, 6(7): e173. doi:10.1371/journal.pbio.0060173
\bibitem{frazzoli}
S.\ Karaman and E. Frazzoli, 2011, ``High-speed Flight through an Ergodic Forest,"
{\em MIT Preprint.}
\bibitem{Justh}
E.W. Justh and P.S. Krishnaprasad, 2006.  ``Steering laws for motion camouflage,'' {\em Proc. R. Soc. A}, 8 December, 2006, vol. 462 no. 2076, pp. 3629-3643. doi: 10.1098/rspa.2006.1742
\bibitem{Kaiser}
M.K.\ Kaiser and L.\ Mowafy, 1993. ``Optical specification of time-to-passage: Observers' sensitivity to global tau,'' Journal of Experimental Psychology: Human Perception and Performance, Vol 19(5), 1028-1040. doi: 10.1037/0096-1523.19.5.1028
\bibitem{Lee-Reddish}
D.N.\ Lee and P.E.\ Reddish, 1981. ``Plummeting gannets: a paradigm of ecological optics,'' {\em Nature}, 293:293-294.
\bibitem{Lon-Praz.1980}
H.C. Longuet-Higgins and K. Prazdny, 1980, ``The Interpretation of a Moving Retinal Image,"
{\em Proceedings of the Royal Society of London.}
\bibitem{Mizutani}
A. Mizutani, J.S. Chahl and M.V. Srinivasan, 2003. ``Insect behaviour: Motion camouflage in dragonflies,'' {\em Nature}, Volume 423 Number 6940,  p. 604.
\bibitem{Srinivasan}
M.V. Srinivasan, 2011. ``Visual control of navigation in insects and its relevance for robotics ,''
{\em Current Opinion in Neurobiology}, 
Volume 21 Number 4 Pages 535--543, doi: 10.1016/j.conb.2011.05.020
http://www.sciencedirect.com/science/article/pii/S0959438811000882
\bibitem{Wang-Frost}
Y.\ Wang and B.J. Frost, 1992. ``Time to collision is signal;ed by neurons in the nucleus rotundus of pigeons,'' {\em Nature}, 356:236-238.
\bibitem{Wed-Poc.2008}
A. Wedel and T. Pock and C. Zach and H. Bischof and D. Cremers, 2008, ``An Improved Algorithm for TV-L1 Optical Flow,"
{\em Revised papers int. Dagstuhl seminar on statistical and geometrical approaches to visual motion analysis.}
\bibitem{Wed-Bro.2011}
A. Wedel and T. Brox and T. Vaudrey and C. Rabe and U. Franke and D. Cremers, 2011, ``Stereoscopic Scene Flow Computation for 3D Motion Understanding,''
{\em International Journal of Computer Vision.}


\end{thebibliography}
\end{document}